%% file: main.tex
\documentclass[a4paper,UKenglish,cleveref, autoref, thm-restate]{lipics-v2021}

\title{Graph Threading with Turn Costs}

\author{Erik D. {Demaine}}{Computer Science and Artificial Intelligence Lab, Massachusetts Institute of Technology, USA \and \url{https://erikdemaine.org/}}{edemaine@mit.edu}{https://orcid.org/0000-0003-3803-5703}{}

\author{Yael {Kirkpatrick}}{Department of Mathematics, Massachusetts Institute of Technology, USA \and \url{https://yaelkirk.github.io/}}{yaelkirk@mit.edu}{https://orcid.org/0009-0007-6718-7390}{NSF Graduate Research Fellowship under Grant No. 2141064.}

\author{Rebecca {Lin}}{Computer Science and Artificial Intelligence Lab, Massachusetts Institute of Technology, USA \and \url{https://rebeccayelin.github.io/}}{ryelin@mit.edu}{https://orcid.org/0000-0003-4747-4978}{MIT Stata Family Presidential Fellowship.}

\authorrunning{Erik D. Demaine, Yael Kirkpatrick, and Rebecca Lin} 

\Copyright{Erik D. Demaine, Yael Kirkpatrick, and Rebecca Lin} 



\acknowledgements{We thank Lily Chung and Timothy Gomez for their helpful comments.}

\category{} 

\EventEditors{}
\EventNoEds{1}
\EventLongTitle{}
\EventShortTitle{}
\EventAcronym{}
\EventYear{}
\EventDate{}
\EventLocation{}
\EventLogo{}
\SeriesVolume{}
\ArticleNo{}
\hideLIPIcs
\nolinenumbers

\input{sections/definitions}

\begin{document}

\maketitle

% body
\input{sections/0-abstract}
\input{sections/1-introduction}
\input{sections/2-preliminaries}
\input{sections/3-hardness}
\input{sections/4-special}
\input{sections/5-conclusions}

% bib
\bibliography{bib}

\end{document}

%% file: sections/definitions.tex
\usepackage{xspace}
\newcommand{\threading}{\textsc{Threading}\xspace}
\newcommand{\decision}{\textsc{Decision Threading}\xspace}
\newcommand{\double}{\textsc{Double Threading}\xspace}
\newcommand{\exactlydouble}{\textsc{Exactly-Double Threading}\xspace}
\newcommand{\perfect}{\textsc{Perfect Threading}\xspace}
\newcommand{\hamcycle}{\textsc{Hamiltonian Cycle}\xspace}
\newcommand{\oneinthreesat}{\textsc{1-in-3 SAT}\xspace}

\newcommand{\tsp}{\textsc{Traveling Salesperson}\xspace}

\newcommand{\true}{\textsc{true}\xspace}
\newcommand{\false}{\textsc{false}\xspace}

\newcommand{\NP}{\textsc{NP}\xspace}
\renewcommand{\P}{\textsc{P}\xspace}

\newcommand{\defn}[1]{\textbf{\textit{\boldmath #1}}}

\newcommand{\set}[1]{\left\{ #1 \right\}}

\newcommand{\clause}[1]{\mathbbm{1}(#1)}

\usepackage{algorithm,algpseudocode}
\usepackage{mathtools}
\usepackage{stackrel}
\usepackage{graphicx,psfrag}
\usepackage{bbm}

\usepackage{xcolor}
\usepackage{soul}

\DeclareUrlCommand\path{\urlstyle{same}}

%% file: sections/0-abstract.tex
\begin{abstract}
How should we thread a single string through a set of tubes so that pulling the string taut self-assembles the tubes into a desired graph? While prior work [ITCS 2024] solves this problem with the goal of minimizing the length of string, we study here the objective of minimizing the total turn cost. The frictional force required to pull the string through the tubes grows exponentially with the total absolute turn angles (by the Capstan equation), so this metric often dominates the friction in real-world applications such as deployable structures. We show that minimum-turn threading is NP-hard, even for graphs of maximum degree $4$, and even when restricted to some special cases of threading. On the other hand, we show that these special cases can in fact be solved efficiently for graphs of maximum degree $3$, thereby fully characterizing their dependence on maximum degree. We further provide polynomial-time exact and approximation algorithms for variants of turn-cost threading: restricting to threading each edge exactly twice, and on rectangular grid graphs.
\end{abstract}

%% file: sections/1-introduction.tex
\section{Introduction}

Deployable structures are dynamic systems capable of transitioning
between multiple geometric configurations,
with applications to compact storage, ease of transportation, and versatility.
One intriguing approach to building deployable structures is to thread a
collection of tubes with string so that different configurations
result from adjustments to string tension.
For example, this idea is at the heart of the ``push puppet'' family of toys.
Interest in this approach is growing within various applied fields,
including fabrication~\cite{Lin2024push} and robotics~\cite{bern2022contact},
with the goal of combining the advantages in adaptability and safety of
soft materials with the precision and integrity of rigid
structures~\cite{bern2023fabrication,patterson2023safe}. 

Threading was formalized into theoretical computer science
in recent work by Demaine, Kirkpatrick, and Lin~\cite{demaine2024graph}.
Specifically, a \defn{threading} of a graph is a closed walk that,
at every vertex, induces a connected graph on the incident edges.
(See Section~\ref{sec:preliminaries} for precise definitions.)
Thus routing a single string along this walk, and pulling the string taut,
self-assembles the tubes (edges) to mimic the connectivity of the given graph.
Past work \cite{demaine2024graph}
gave a polynomial-time algorithm for finding the
\emph{minimum-length} threading of a given graph.

In this paper, we analyze threading according to a \emph{turn} metric,
motivated by challenges that arise when constructing threadings in
practice~\cite{Lin2024push}. 
% \todo{In this setting we assign a turning cost function $c: E \times E \mapsto \mathbb{R}_{+}$, 
% and we define the turn cost of a threading $T$ to be TODO.} 
While minimizing length makes sense to save on material cost and
manufacturing time, the feasibility of deployable structures is
limited by the force required to pull the string, which is usually
dominated by the amount of friction in the system.
According to the Capstan equation, friction increases \emph{exponentially}
with the sum of the absolute values of turn angles in the threading
route~\cite{Kyosev2015braiding}.
%Friction plays a major role in determining the force required
%to draw the string through the system.
This motivates us to find threadings of minimum total ``turn cost'',
where we allow each vertex to specify the (symmetric) cost of turning
between every pair of incident edges.
%In a geometric setting, we care about the total absolute turn angle
%of the walk, so the turn cost at a vertex is the absolute turn angle.
%In the graphical setting studied here, we allow each vertex to specify
%the (symmetric) cost of turning between every pair of incident edges.
(Again see Section~\ref{sec:preliminaries} for precise definitions.)

\subsection{Our Results}

In this paper, we analyze the \threading~problem under a turn metric, where the goal is to find the threading of a graph with minimum total turn cost. We present several new results: 

\begin{itemize}
\item In Section~\ref{sec:hardness}, we prove \NP-completeness of
  \threading~under a turn metric in various settings.
  These results provide a stark contrast to the length metric where \threading~can be solved in polynomial time~\cite{demaine2024graph}.

  \begin{itemize}
  \item
    We give a simple proof that \threading~is \NP-complete, even if the objective is simply to minimize the number of turns, that is, the turn costs are either $0$ or $1$.
    This reduction shows how just one linear-degree vertex can
    effectively simulate an arbitrary Hamiltonian Cycle instance,
    motivating the restriction to bounded degree.
  \item
    We prove that \threading~is \NP-complete even for graphs of
    maximum degree $4$, via a reduction from \oneinthreesat.
  \item
    \perfect~is a special case where we require using
    the smallest possible turn cost when restricted to every vertex.
    Such a threading is not always possible, and we prove that it is
    \NP-complete to decide even for graphs of maximum degree~$4$.
    The analogous problem for the length metric (visiting each vertex the
    fewest possible times) is of course polynomial, and was a helpful
    stepping stone toward the full polynomial-time algorithm
    \cite{demaine2024graph}.
  \item
    It follows from our reductions that both {\threading} and {\perfect}
    are \NP-complete even when we further constrain the threading to
    at most two visits per edge, a problem we call \double. 
  \end{itemize}
\item 
  In Section~\ref{sec:max-degree-3}, we provide polynomial-time algorithms for {\perfect} and {\double}
    in graphs of maximum degree $\leq 3$.
    As we show these problems \NP-hard for graphs of maximum degree $4$, our results give a complete characterization of the complexity of these two problems.
    (It remains open whether the more general \threading~problem is hard
    in graphs of maximum degree $\leq 3$.)
  \item 
  In Section~\ref{sec:special}, we further analyze interesting special cases of
  {\threading} under our turn metric, in particular developing several more
  polynomial-time algorithms:
  \begin{itemize}
  \item 
    When each edge must be traversed \emph{exactly} twice, \exactlydouble, we show that the problem is equivalent to solving a sequence of \textsc{Traveling Salesman} problems, one at each vertex. We use this fact to prove the intractability of {\exactlydouble} for unbounded-degree graphs, to develop polynomial-time approximation algorithms, and to develop exact algorithms for graphs of logarithmic degree or for minimizing the number of turns (turn costs $\in \{0,1\}$).
  \item We present a strategy for threading on rectangular grid graphs.
  \item 
    If the ratio between the maximum and minimum turn costs is~$r$,
    we provide a $2r$-approximation algorithm.
  \end{itemize}

% \item In Appendix~\ref{sec:grid}, we present an algorithm for \threading~on rectangular grid graphs.
\end{itemize}

Minimizing turn cost is a generalization of minimizing length
(e.g., define the $e \to f$ turn cost to be $\frac{|e|+|f|}2$),
but minimum-turn threadings can have arbitrarily longer-than-optimal lengths
(see \autoref{fig:lower-bound}).

\subsection{Related Work}

Turn costs have been studied for several other problems.
In the angular-metric traveling salesman problem~\cite{aggarwal2000angular},
the goal is to minimize the total absolute turn angle of the tour. 
In minimum-turn milling \cite{arkin2005optimal}, the goal is to cover a region
with a moving tool while minimizing the total turn cost or number of turns,
motivated by machinery such as mills which operate more efficiently
moving in straight lines rather than turning.
In minimum-turn Eulerian tour \cite{ellis2015dna}, the goal is to minimize the
total turn cost; in surprising contrast to regular Eulerian tours,
this problem is NP-hard even for minimizing the number of turns
(turn costs $\in\{0,1\}$) in graphs of bounded degree.

%% file: sections/2-preliminaries.tex
\section{Preliminaries}
\label{sec:preliminaries}

In this section, we discuss the necessary background for \threading~ and introduce several special variants of the problem. We assume all graphs in question to be connected with a minimum degree of at least $2$. 

\begin{definition}[Threading \cite{demaine2024graph}]
\label{def:threading}
    Let $G = (V, E)$ be a graph with $n = |V|$ vertices and $m=|E|$ edges, 
    where each edge $e \in E(G)$ represents a tube and each vertex $v \in V(G)$ represents the junction where the tubes incident to $v$ meet. A \defn{threading} $T$ of $G$ is a closed walk through $G$ that visits every edge at least once, induces connected ``junction graphs'' at every vertex, and has no ``U-turns'', meaning that the walk cannot traverse the same edge twice in sequence. The \defn{junction graph} $J(v)$ of a vertex $v$ induced by a closed walk has a vertex for each tube incident to $v$, and has an edge between two of these vertices every time the walk visits $v$ immediately in between traversing the corresponding tubes. 
\end{definition}

\begin{figure*}[h]
    \centering \includegraphics[width=\textwidth]{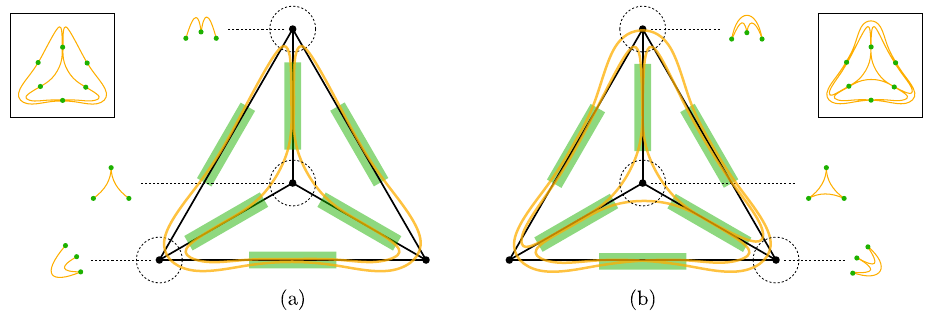}
      \caption{\label{fig:junction-graphs} 
      Two distinct threadings of a tetrahedron, where each junction graph is either a tree (a) or a cycle (b) (\cite[Figure 2]{demaine2024graph}) Their threading graphs are shown in the corners.}
\end{figure*}

While prior work~\cite{demaine2024graph} focused on minimum-\emph{length} threading, we shift our attention to minimum-\emph{turn} threading. Let us first define turn costs formally: 

\begin{definition}[Turn Cost]
\label{def:turn}
A \defn{turn} at a vertex $v$ is an unordered pair $(uv, vu)$ of distinct edges both incident to $v$.
Let $T(G)$ denote the set of all turns in $G$. 
\defn{Turn costs} for $G$ are defined by a function $\alpha: T(G) \rightarrow \mathbb{R}_{\geq 0}$, specifying a nonnegative real number for each turn in $G$.
For a closed walk $W$ in $G$ traversing a sequence of vertices $v_1, \dots, v_{\ell}$, the \defn{turn cost} of $W$ is the sum of the turn costs for all consecutive pairs of edges (turns) along $W$:
\begin{equation*}
\label{eqn:turn}
\alpha(W) = 
\sum_{i=2}^{\ell-1} \alpha(v_{i-1}v_{i}, v_{i}v_{i+1})
+ \alpha(v_{\ell-1}v_{\ell},v_{\ell}v_{1})
+ \alpha(v_{\ell}v_1, v_1v_2).
\end{equation*}
\end{definition}

Now we can state the problem of finding a minimum-turn threading as: 

\begin{problem} [\threading]
\label{prob:threading}
Given $(G, \alpha)$, find a threading $T$ of $G$ that minimizes $\alpha(T)$.
\end{problem}

The decision version of \threading is as follows: 

\begin{problem}[\decision]
\label{prob:decision-threading}
Given $(G, \alpha)$ and a nonnegative real number $c$, does there exist a threading $T$ of $G$ such that $\alpha(T) \leq c$?
\end{problem}

\subparagraph*{Na\"ive Double-Threading Solution.} We can na\"ively compute a threading. First, assign each junction graph to be a cycle; see the example in Figure~\ref{fig:junction-graphs}b. Then, compute an Euler tour in $O(m)$ time through the union of these junction graphs (Figure~\ref{fig:junction-graphs}) that prohibits U-turns~\cite[Section 2.2]{demaine2024graph}. This tour is a threading that ``double-threads'' each edge. 

\begin{definition}[Minimum-Turn Perfect Threading]
A \defn{perfect threading} is a threading where every junction graph is a tree. Not all graphs have perfect threadings~\cite{demaine2024graph}.
A \defn{minimum-turn perfect threading} is a threading where each junction graph is a minimum spanning tree. 
\end{definition} 

By definition, each junction graph induced by a threading is connected, and so contains a spanning tree. Hence, the cost incurred by the threading at each junction must be at least that of a minimum spanning tree at that junction. Therefore, if a minimum-turn perfect threading exists, it is a minimum-turn threading. It follows that finding a minimum-turn perfect threading is a subproblem of \threading.

\begin{problem}[\perfect]
Given $(G, \alpha)$, 
find a minimum-turn perfect threading of $G$, if it exists. 
\end{problem} 

At times, it may be pragmatic to limit the number of traversals through a tube, which motivates our final problem definition: 

\begin{problem}[\double]
\label{prob:double-threading}
Given $(G, \alpha)$, find a minimum-turn threading of $G$ that traverses each edge at most twice. 
\end{problem}

We prove our first result: \decision is in \NP. It is straightforward to verify that the same conclusion holds for the decision versions of \perfect and \double. 

\begin{theorem}
\label{thm:np}
    \decision is in \NP. 
\end{theorem}

\begin{proof}

A purported solution $T$ to \decision can be verified by checking that $T$ is indeed a threading of $G$ and that $\alpha(T) \leq c$. Each of these checks can be performed in $O(|T|)$ time. 

We now prove that, for any threading $T$, there exists a threading $T'$ with size polynomial in $m$ such that $\alpha(T') \leq \alpha(T)$, thereby showing that a polynomial-size witness for the solution exists. We can construct $T'$ by repeatedly deleting cycles in $T$ that do not disconnect any junction graphs until no such cycles remain.  We next show that $T'$ traverses each edge at most $O(m^2)$ times.

Assume for a contradiction that $T'$ traverses an edge $uv \in E(G)$ more than $|T(G)|+1$ times, where $T(G)$ denotes the set of all turns in $G$. Then $T'$ contains at least $|T(G)|+1$ cycles, each formed via a revisit to $uv$. By the pigeonhole principle, there exists at least one cycle that fails to traverse a turn not already covered by the other cycles. This redundant cycle may be removed to produce a shorter walk for which junction graphs are still connected. This contradicts our construction of $T'$. Hence, it follows that $uv$ is traversed $\leq |T(G)|+1$ times. Since $|T(G)| = \sum_{v \in V(G)}d(v)^2 \in O(m^2)$~\cite{de1998upper}, our conclusion follows. 
\end{proof}

While the proof of Theorem~\ref{thm:np} demonstrates that an edge is traversed $O(m^2)$ times, Figure~\ref{fig:lower-bound} shows an example where an edge is traversed $\Omega(m)$ times.

\begin{figure*}[h]
\centering
\includegraphics[width=\textwidth]{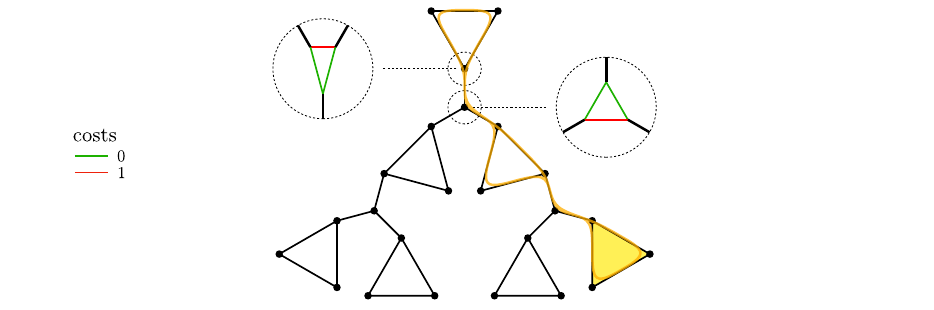}
\caption{\label{fig:lower-bound} 
A tree-like graph with degree-3 junctions as marked by the circles. A green line connecting an edge pair indicates a turn cost of $0$, whereas a red line indicates a turn cost of $1$. All other turns have zero cost. The minimum-turn threading of this graph deviates far from a perfect threading --- it traverses the ``root'' triangle as many times as there are ``leaf'' triangles.}
\end{figure*}

%% file: sections/3-hardness.tex
\section{Intractability of \threading}
\label{sec:hardness}

In Section~\ref{sec:general-hardness}, we show that \threading is \NP-complete, even with turn costs $\in \set{0,1}$, via a simple reduction from the \hamcycle problem. The reduction uses the turn costs at a single linear-degree vertex to encode the connectivity of an arbitrary graph. We later utilize this technique to give positive results in Section~\ref{sec:tsp}. In Section~\ref{sec:bounded}, we show that \threading remains NP-complete for graphs of maximum degree $4$ by simulating \oneinthreesat with vertices of degree $\leq 4$. 

\subsection{General Graphs}
\label{sec:general-hardness}

\begin{theorem}
    \threading~is \NP-complete, even if only minimizing the number of turns. 
\end{theorem}

We prove the hardness of \threading by reducing from \hamcycle, which is \NP-hard~\cite{garey1979computers}. The \hamcycle~problem asks whether a graph $G$ contains a Hamiltonian cycle, a cycle that visits every vertex exactly once. 

Given an instance $G$ of \hamcycle~with $n$ vertices and $m$ edges, we construct an instance $(H, \alpha)$ of \threading~as follows (Figure~\ref{fig:ham-cycle}): First, create $n$ copies of the cycle graph $C_3$, one per vertex in $G$, and connect each cycle via a bridge to a central vertex $c$. We refer to each cycle and its incident bridge as an ``arm'' and denote the arm corresponding to vertex $v_i$ as $a_i$. Next, at vertex $c$, assign each pair of incident edges from distinct arms $a_i$ and $a_j$ a turn cost of $0$ if $v_iv_j \in E(G)$. Assign all remaining turns a cost of $1$. By design, traversing each arm incurs a cost of $4$. The time to construct $(H, \alpha)$ is $O(m+n)$. 

\begin{figure*}[h]
    \centering \includegraphics[width=\textwidth]{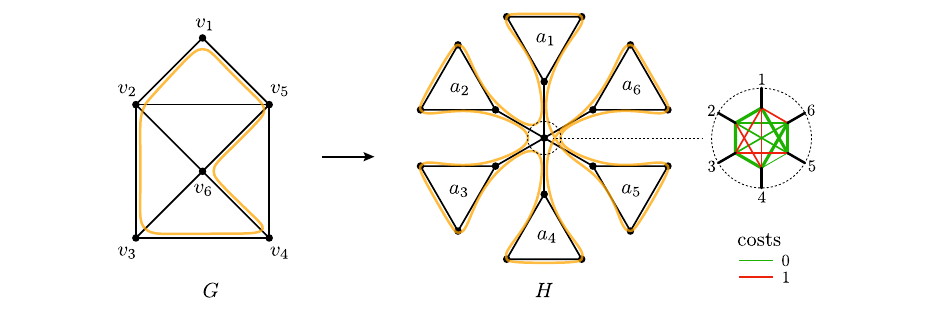}
      \caption{\label{fig:ham-cycle} 
    Constructing an instance $H$ of \threading from an instance $G$ of \hamcycle. The turn costs at the central vertex in $H$ are depicted inside the right-most circle, where a green line connecting a pair of edges indicates a turn cost of $0$ and a red line indicates a turn cost of $1$. Turn costs for unlabeled turns are assumed to be $1$. The edges of the junction graph induced by the threading in orange are in bold.}
\end{figure*}

\begin{lemma}
\label{lem:ham-cycle}
Graph $G$ contains a Hamiltonian cycle if and only if its corresponding instance $(H, \alpha)$ to \threading~has a threading of turn cost $4n$. 
\end{lemma}

\begin{proof}

Suppose $G$ has a Hamiltonian cycle $C$. Label the vertices of $G$ such that $C$ is given by the sequence of vertices $v_1, v_2, . . . , v_n, v_1$. Construct a walk $T$ through $H$ that traverses the arms the same order: $a_1, a_2, \dots, a_n, a_1$. Here, $a_i a_j a_k$ indicates the walk entering arm $a_j$ directly from arm $a_i$, traversing all edges of $a_j$ in either a clockwise or counter-clockwise fashion, and then exiting $a_j$ directly to arm $a_k$. This walk creates a connected junction graph at each vertex. Most notably, at the central vertex $c$, the junction graph is given by the Hamiltonian cycle $C$ where each vertex $v_i$ maps to the bridge from arm $a_i$. It follows that $T$ is a threading of $H$. This threading incurs a turn cost of $4n$: Each arm of cost $4$ is traversed exactly once because $C$ is a Hamiltonian cycle, and moving between consecutive arms $a_i$ and $a_j$ in $H$ incurs zero cost because $v_iv_j$ is in the Hamiltonian cycle and is hence in $E(G)$. 

Now assume there exists a threading $T$ of $H$ of turn cost $4n$. By definition, $T$ visits each arm at least once, and upon entering an arm, it fully traverses the arm because U-turns are prohibited. Given that each arm incurs a cost of $4$, turn cost $\alpha(T)=4n$ implies two facts: (1) $T$ traverses each arm \textit{exactly} once and (2) incurs zero turn cost at vertex $c$. Labeling the sequence of arms traversed by $T$ as $a_1, a_2, \dots, a_n$, we map $T$ to vertices $v_1, v_2, \dots, v_n, v_1$ in $G$. Consecutive vertices in this sequence are adjacent by (2), hence this sequence forms a cycle. Furthermore, each vertex is visited once per (1), so this cycle is Hamiltonian. 
\end{proof}

It follows from Theorem~\ref{thm:np} and Lemma~\ref{lem:ham-cycle} that \threading is \NP-complete, even with turn costs $\in \set{0,1}$. 

\subsection{Maximum-Degree-4 Graphs}
\label{sec:bounded}

Next we consider bounded degree graphs and show that even the special case of \perfect~is \NP-complete in this setting.

\begin{theorem}\label{thm:hardboundeddegree}
     \perfect~in graphs of maximum degree $4$ is \NP-complete.
\end{theorem}

Since \perfect~is a special case of \threading (Section~\ref{sec:preliminaries}), we conclude: 

\begin{corollary}
    \threading~in graphs of maximum degree $4$ is \NP-complete. 
\end{corollary}

We prove Theorem~\ref{thm:hardboundeddegree} via a reduction from \oneinthreesat, which is \NP-hard~\cite{garey1979computers}. 

\begin{definition}[\oneinthreesat]
    Given $m$ clauses consisting of three literals each, over a set of $n$ boolean variables, the \oneinthreesat~problem asks to determine whether there exists a truth assignment to the variables such that each clause contains exactly one \true~literal.
\end{definition}

For example, the formula 
\begin{equation}
\label{eqn:sat-formula}
    \varphi = 
    \clause{w, x, \lnot y} \wedge 
    \clause{\lnot w, x, \lnot z} \wedge 
    \clause{x, y, z},
\end{equation}
where $\clause{x,y,z}$ represents that exactly one of $x,y,z$ is \true,
has a satisfying assignment of $x = \false, y=\true, z=\false $ and $ w =\true$.

Given an arbitrary instance $\varphi$ of \oneinthreesat, we construct an instance of \perfect. This graph comprises several gadgets to simulate $\varphi$: a ``variable gadget'' for each variable, a ``clause gadget'' for each clause, and ``split gadgets'' to handle the occurrence of the same literal in multiple clauses. 

\subparagraph*{Variable Gadget.} 
For each variable $x$, construct a degree-$3$ junction as shown in \autoref{fig:var_gadget}a. A turn connected by a green line indicates a turn cost of 1 and a red line indicates a turn cost of 2. The left and right edges incident to the junction represent the literals $x$ and $\lnot x$, respectively. A key observation is that any threading that forms a minimum spanning tree (MST) at this junction will single-thread the top edge and select one of the bottom edges to single-thread and the other to double-thread. We use this double-threading to represent the assignment of $x$: If the bottom left edge is double-threaded, then $x$ is assigned \true; conversely, if the bottom right edge is double-threaded, then $x$ is assigned \false, that is, $\lnot x$ is assigned \true (\autoref{fig:var_gadget}). 

\begin{figure*}[ht]
    \centering
        \includegraphics[width=\textwidth]{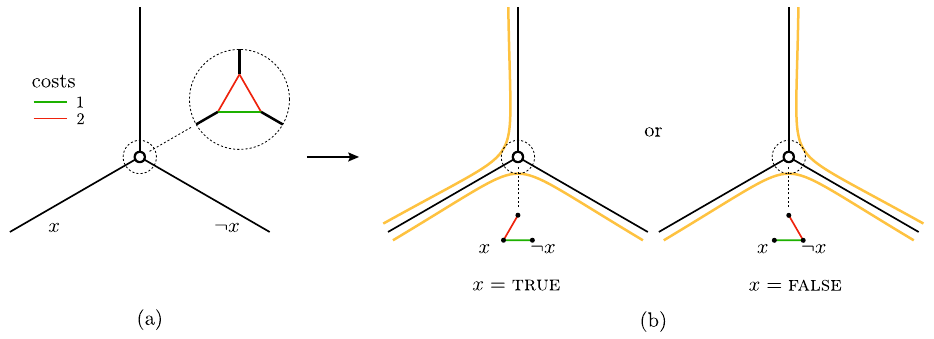}
          \caption{\label{fig:var_gadget} (a) The junction representing variable $x$, with green indicating a cost of $1$ and red indicating a cost of $2$. (b) Two possible minimum-cost threadings of the junction and their corresponding junction graphs; double-threading the left and right edges of the junction corresponds to setting $x$ and $\lnot x$ to \true, respectively.}
    \end{figure*}

\subparagraph*{Clause Gadget.} To simulate a clause $\clause{x,y,z}$, we construct a junction that is adjacent to variable gadgets $x$, $y$, and $z$ or their copies (as can arise from the use of the splitting gadget), with all turn costs being $1$  (\autoref{fig:clause_gadget}). Any threading that induces an MST on this junction will double-thread exactly one of the incident edges and single-thread the remaining two. By using the double-threading to represent a \true assignment, this corresponds to having exactly one literal out of $x,y,z$ assigned \true. 

\begin{figure*}[h]
    \centering
        \includegraphics[width=\textwidth]{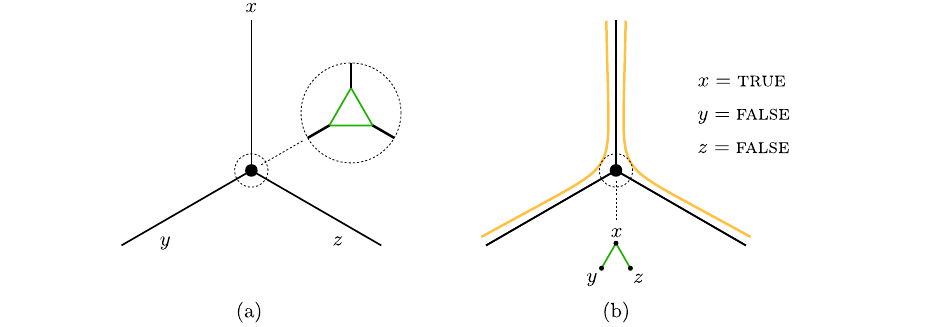}
          \caption{\label{fig:clause_gadget} (a) The gadget for clause $\clause{x, y, z}$ and (b) an example threading that corresponds to assigning $x$ to \true.}
    \end{figure*}

    \subparagraph*{Split Gadget.} This final and most involved gadget allows us to ``split'' an edge representing a literal $x$. It transforms a single edge simulating $x$ into three edges, all simulating $x$ and thus receiving the same assignment. The gadget is depicted in \autoref{fig:split_gadget}, where the green, red, and blue lines indicate turns with costs of $1$, $2$, and $3$, respectively. 

        \begin{figure*}[h]
        \centering
          \includegraphics[width=\textwidth]{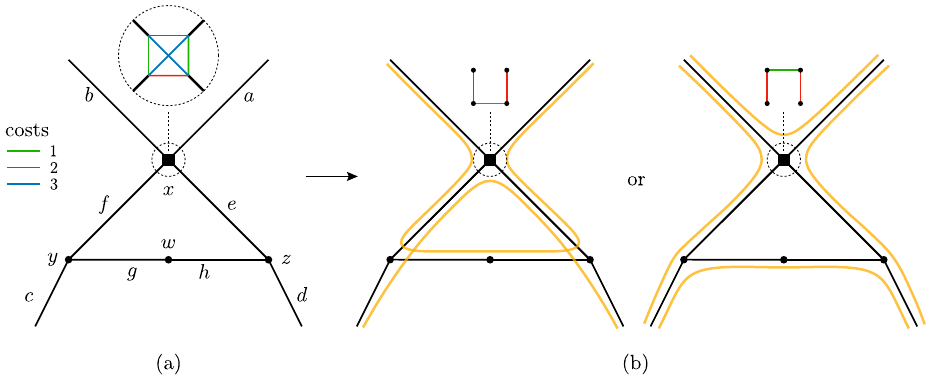}
          \caption{\label{fig:split_gadget} A split gadget. All unmarked turns have unit cost.}
    \end{figure*}

    \begin{claim} \label{clm:splitting_gadget}
        Any \perfect~of the split gadget (\autoref{fig:split_gadget}) threads edges $a,b,c,d$ the same number of times. 
    \end{claim}

    \begin{proof}
        First, note that the junction graph at vertex $x$ must be one of the two MSTs shown in \autoref{fig:split_gadget}b. Therefore, under any threading that induces an MST for $J(x)$, either edges $a$ and $b$ are single-threaded while $e$ and $f$ are double-threaded, or vice versa. Next, consider vertex $w$. To induce an MST on $w$, edges $g$ and $h$ must be single-threaded. Hence, by simply counting the number of traversals through vertex $y$, if edge $f$ is single-threaded, then edge $c$ must be threaded twice, and vice versa. The same conclusion applies symmetrically to edges $e$ and $d$. We conclude that a perfect threading of the split gadget follows one of two options depicted in \autoref{fig:split_gadget}b. In both scenarios, the edges $a,b,c$, and $d$ are all threaded the same number of times. 
    \end{proof}

   \subparagraph*{Full Construction.} Given a \oneinthreesat~instance $\varphi=\clause{\ell_{1,1},\ell_{1,2}, \ell_{1,3}} \land \clause{\ell_{2,1},\ell_{2,2},\ell_{2,3}} \land \dots \land \clause{\ell_{m,1},\ell_{m,2},\ell_{m,3}}$ using variables $x_1, \ldots, x_n$, construct the graph $G(\varphi)$ as follows: 

   \begin{algorithm}
    \caption{Constructing a (multi)graph $G(\varphi)$ with turn costs from a \oneinthreesat formula $\varphi$.}
    \label{alg:reduction}
   \begin{enumerate}
       \item 
       For every variable $x_i$, construct a variable gadget as shown in Figure~\ref{fig:var_gadget}. We refer to the center-top edge as the ``edge simulating variable $x_i$'', and the bottom-left edge as an ``edge simulating \emph{literal} $x_i$'' and the bottom-right edge as an ``edge simulating \emph{literal} $\lnot x_i$''. 
       
       \item 
       If literal $x_i$ appears more than once in $\varphi$, construct a split gadget and connect one of its outgoing edges to the edge simulating $x_i$. Treat the remaining three outgoing edges of the gadget as ``edges simulating literal $x_i$''. Repeat this process until the number of edges simulating $x_i$ is at least the number of appearances of $x_i$ in $\varphi$. 

       \item For every clause $\clause{\ell_{i, 1}, \ell_{i, 2}, \ell_{i,3}}$, construct a clause gadget and connect it to edges simulating $\ell_{i, 1}, \ell_{i, 2}$, and $\ell_{i,3}$. 

       \item 
       \label{step:connection}
       Duplicate our entire construction up to this point. Connect each dangling edge simulating a variable or a literal directly to its copy. 
   \end{enumerate}
    \end{algorithm}

 This construction runs in $O(n+m)$ time. An example construction for \autoref{eqn:sat-formula} is provided in Figure~\ref{fig:construction}a.

\begin{figure*}[ht]
\centering
\includegraphics[width=\textwidth]{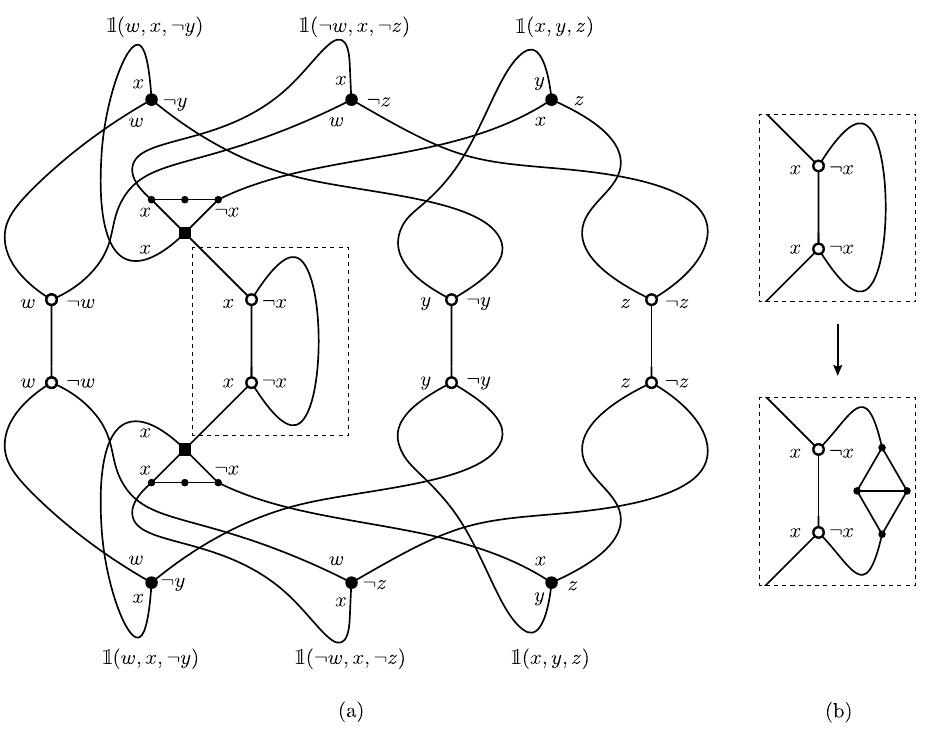}
\caption{\label{fig:construction} (a) Construction of multigraph $G(\varphi)$ for the formula $\varphi = \clause{x,y,z}\land \clause{x, \lnot y, w} \land \clause{x, \lnot z, \lnot w}$. By directly connecting the dangling edges simulating the (unused) literal $\lnot x$ (dashed) in Step~\ref{step:connection} of Algorithm~\ref{alg:deg3perfectthreading}, we create multiedges. (b) We can instead connect the edges via a de-multigraph gadget (\autoref{fig:demultigraphing}) to obtain a simple graph $G'(\varphi)$.}
\end{figure*}

    \begin{lemma}
    \label{lem:sat}
        $G(\varphi)$ has a minimum-turn perfect threading if and only if $\varphi$ has a satisfying assignment.
    \end{lemma}
    \begin{proof}
        Let's first construct a minimum-turn perfect threading of $G(\varphi)$ from a satisfying assignment of $\varphi$ as follows: For each variable $x$, thread its gadget according to Figure~\ref{fig:var_gadget}b, double-threading the edge simulating $x$ if $x = \true$ (left) and the edge simulating $\lnot x$ if $x = \false$ (right). For each literal $\ell$, thread its split gadget(s) according to Figure~\ref{fig:split_gadget}b so that the edges $a,b,c,d$ are either all double-threaded if $\ell = \true$ (right) or all single-threaded if $\ell = \false$ (left). For each clause, thread its gadget as shown in Figure~\ref{fig:clause_gadget}b so that the edge simulating its single true literal is double-threaded, while the other two edges are single-threaded.      
        By design, this threading induces an MST at each junction of $G(\varphi)$ so it is a minimum-turn perfect threading. 
    
        Now let's demonstrate the converse.  Suppose we have a minimum-turn threading of $G(\varphi)$. For each variable $x$, if the edge simulating $x$ in its gadget is double-threaded, then assign $x$ to \true, and \false, otherwise. We claim that this assignment satisfies $\varphi$. Consider a clause $\clause{\ell_{i, 1}, \ell_{i, 2}, \ell_{i,3}}$ in $\varphi$. Under a perfect threading, the gadget representing this clause must have an edge double-threaded and two single-threaded. Since these edges represent the literals $\ell_{i, 1}, \ell_{i, 2}, \ell_{i,3}$, this means that exactly one of these literals is assigned \true, thereby satisfying $\clause{\ell_{i, 1}, \ell_{i, 2}, \ell_{i,3}}$.
    \end{proof}

\begin{figure*}[h]
\centering
\includegraphics[width=\textwidth]{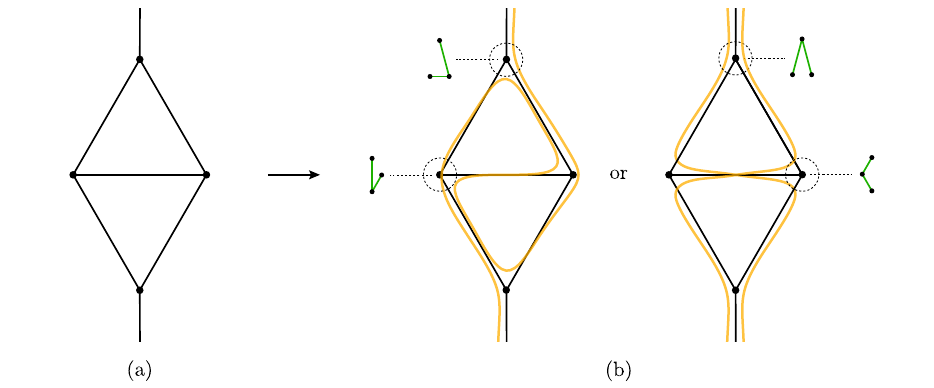}
\caption{\label{fig:demultigraphing} A de-multigraph gadget. All turns are of unit cost.}
\end{figure*}

\subparagraph*{De-Multigraph.} 
As shown in \autoref{fig:construction}a, if a literal does not appear in $\varphi$, it creates a multiedge with the edge simulating its corresponding variable. To address this issue and transform the multigraph $G(\varphi)$ into a simple graph $G'(\varphi)$, we modify Step~\ref{step:connection} of the construction by connecting the dangling edge simulating the literal to its copy using the ``de-multigraph gadget'' shown in \autoref{fig:demultigraphing}, rather than directly (\autoref{fig:construction}b). This gadget can be threaded in a way that induces an MST at each of its junctions, regardless of whether the incoming literal edges are threaded once or twice; see \autoref{fig:demultigraphing}b. Hence, we can adapt the proof of Lemma~\ref{lem:sat} and conclude that $G'(\varphi)$, with the specified turn costs, qualifies as an instance of \perfect~and has a minimum-turn perfect threading if and only if the \oneinthreesat instance $\varphi$ has a satisfying assignment. 
    
This result concludes our proof of Theorem~\ref{thm:hardboundeddegree}. As each edge in graph $G'(\varphi)$ is traversed no more than twice, our reduction further implies the hardness of \double. Hence: 

\begin{corollary}
\double~is NP-complete for graphs of maximum degree $4$. 
\end{corollary}

%% file: sections/4-special.tex
\section{Maximum-Degree-3 Graphs}
\label{sec:max-degree-3}

Here we present polynomial-time algorithms for solving \perfect and \double on graphs of maximum degree $3$, thereby fully characterizing the complexity of these two problems. 

\subsection{\perfect}

\begin{theorem}
    \perfect~in graphs of maximum degree $3$ is in \P.
\end{theorem}

Let us consider an instance $(G, \alpha)$ of \perfect with graph $G$ of maximum degree $3$. We begin our proof of the theorem by observing that a spanning tree at a degree-$3$ junction $v \in V(G)$ is a path of length $2$. Hence, one of the edges incident to $v$ must be double-threaded while the remaining two edges are single-threaded. From this observation, we deduce that finding a perfect threading in $G$ is a matching problem: Selecting one incident edge per degree-$3$ vertex to double-thread is equivalent to finding a set of edges such that each degree-$3$ vertex is incident to one edge from this set. 

However, not every spanning tree qualifies as a \emph{minimum} spanning tree (MST). Hence, to achieve a \textit{minimum-turn} perfect threading, we must double-thread only those edges that could induce MSTs. That is, we seek matchings that select for edges $uv$ corresponding to degree-$2$ vertices in MSTs for both junction graphs $J(u)$ and $J(v)$. See our proposed algorithm below: 

\begin{algorithm}
    \caption{\perfect~on maximum-degree-3 graphs.}
    \label{alg:deg3perfectthreading}
    \begin{enumerate}
        \item 
        \label{step:one}
        Delete all degree-$2$ vertices in $G$ and their incident edges. 
        \item 
        \label{step:two}
        For every remaining vertex $v\in V$, consider all MSTs for $J(v)$ and delete any edge incident to $v$ of degree $1$ in all MSTs. 
        \item 
        \label{step:three}
        Compute a perfect matching of the remaining graph $G'$. Return the matching if it exists. 
    \end{enumerate}
    \end{algorithm}

\subparagraph*{Correctness.}
The correctness the algorithm follows from the claim below: 
    \begin{claim}
        Graph $G$ has a perfect threading if and only if the graph $G'$ produced by Algorithm~\ref{alg:deg3perfectthreading} has a perfect matching.
    \end{claim}
    \begin{proof}
    Suppose $G$ has a minimum-turn perfect threading $T$. Consider the set $M$ of edges that are double-threaded by $T$. Since $T$ induces an MST at every junction, none of the edges in $M$ could be removed in Step~\ref{step:two} of Algorithm~\ref{alg:deg3perfectthreading}. Hence, every edge of $M$ exists in $G'$. Furthermore, since every vertex in $G'$ corresponds to a degree-$3$ vertex in $G$, it must have exactly one incident edge that is double-threaded by $T$, implying one incident edge in $M$. Therefore, $M$ is a perfect matching in $G'$.

   Conversely, given a perfect matching $M$ of $G'$, we can construct junction graphs that result in a minimum-turn perfect threading: For every degree-2 vertex, we take the trivial junction graph comprising a single edge. For every degree-$3$ vertex $v$ with its matching $u$ in $M$, we set $J(v)$ to be the path of length $2$, with the central vertex corresponding to the edge $uv$. This path forms an MST by the construction of $G'$ (Step 2). 
\end{proof}

\subparagraph*{Running-Time.}
Step~\ref{step:one} runs in $O(n)$ time. Step~\ref{step:two} also runs in $O(n)$ time given that each degree-$3$ vertex has at most three MSTs.  Finally, since $|G'| \in O(n)$, Step~\ref{step:three} runs in $O(n^2\log{n})$ time via the perfect matching algorithm by Galil, Micali, and Gabow~\cite{galil1986ev}. 

\subsection{\double}
\label{sec:double}

While perfect threadings may not exist for some graphs, double-threadings are always possible --- refer to our na\"ive algorithm in Section~\ref{sec:preliminaries} for details. Here we present an efficient solution to \double~in maximum-degree-$3$ graphs. 

\begin{theorem}
    \double~in graphs of maximum degree 3 can be solved in time $O(m^2\log m)$. 
\end{theorem}

\begin{proofsketch}

Suppose $G$ is a graph of maximum degree $3$. Let $\mathcal{C}$ be a collection of vertex-disjoint simple cycles in $G$. Prior work~\cite[Lemma 9]{demaine2024graph} achieves a threading $T$ that single-threads every edge in $\mathcal{C}$ and double-threads every other edge. 
Because $G$ has maximum degree $3$, the turn cost of this threading is: 
\[
\alpha(T) = \underbrace{\sum_{(uv, vw) \in T(G)} \alpha(uv, vw)}_{\text{total cost of all turns in $G$}} - \sum_{C \in \mathcal{C}}\alpha(C). 
\]
Therefore, solving \double~on $G$ is equivalent to finding a collection of vertex-disjoint simple cycles in $G$ whose total turn cost is maximized. We give an algorithm to solve this problem; see Algorithm~\ref{alg:vertex-disjoint-cycles} below. 

\begin{algorithm}
\caption{Computing maximum-turn vertex-disjoint simple cycles.}
\label{alg:vertex-disjoint-cycles}
\begin{enumerate}
    \item
    \label{step:construction}
    Construct a weighted graph $G'$ from $(G, \alpha)$
    (Figure~\ref{fig:vertex_disjoint}) as follows: 
    \begin{enumerate}
        \item For each edge $uv$, create a zero-weight path of length $3$ comprising vertices $uv^{+}$, $uv^{-}$, $vu^{-}$, and $vu^{+}$. 
        \item 
        \label{step:turn}
        For each turn $(uv, vw)$ in $T(G)$, add an edge of weight $\alpha(uv, vw)$ connecting the vertices $vu^{+}$ and $vw^{+}$. 
    \end{enumerate}
    \item 
    \label{step:matching}
    Compute a maximum weight perfect matching $M$ in $G'$. 
    \item Return the set of edges $S = \set{uv \in E(G): (uv^{-}, vu^{-}) \in M}$.
\end{enumerate}
\end{algorithm}

\begin{figure*}[h]
\centering
\includegraphics[width=\textwidth]{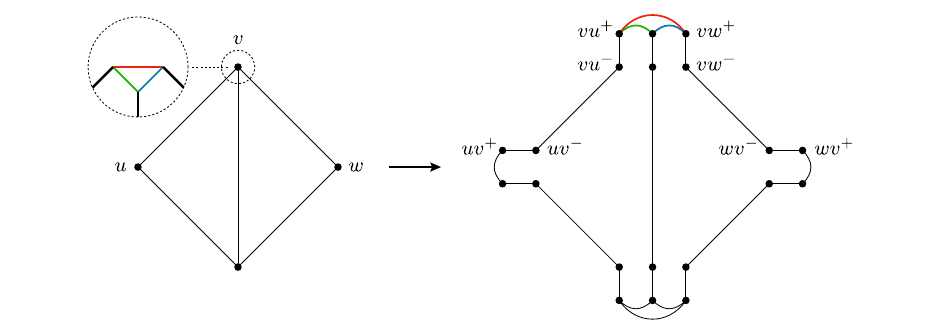}
\caption{\label{fig:vertex_disjoint} Construction of $G'$ from $G$ in Algorithm~\ref{alg:vertex-disjoint-cycles}.}
\end{figure*}

\subparagraph*{Correctness.} We briefly explain why the matching $M$ in $G'$ corresponds to a collection of vertex-disjoint simple cycles in $G$. Consider two possible cases for each vertex $v \in V(G)$: 

\begin{enumerate}
    \item Suppose $M$ contains a weighted edge $(vu^{+}, vw^{+})$ created in Step~\ref{step:turn}. Only one such edge is possible for $v$ given that the set of vertices $\{v{x}^{+}\}_{x\in N(v)}$ are pairwise connected and have size at most $3$. It follows that the remaining vertices $\{v{x}^{+}\}_{x\in N(v)\setminus\{u,w\}}$ are matched with their counterparts in $\{vx^{-}\}_{x\in N(v)\setminus\{u,w\}}$, and so vertices $vu^{-}$ and $vw^{-}$ must be saturated by edges $(vu^{-}, uv^{-})$ and $(vw^{-}, wv^{-})$, respectively. These edges correspond to some cycle traversing edges $uv$ and $vw$ through $u$ in $G$. 
    \item Otherwise, all vertices in $\{vx^{+}\}_{x\in N(v)}$ are matched to vertices in $\{v{x}^{-}\}_{x\in N(v)}$. This indicates that $v$ does not participate in any cycle.
\end{enumerate}

\subparagraph*{Running-Time.} Graph $G'$ has $O(m)$ vertices and edges. We can adapt the minimum-weight perfect matching algorithm by Galil, Micali, and Gabow~\cite{galil1986ev} for Step~\ref{step:matching}. This procedure runs in $O(|V(G')||E(G')| \log|V(G')|) = O(m^2 \log m)$ time, which dominates the $O(m)$ construction time of $G'$ in Step~\ref{step:construction}. Hence, the overall running time of Algorithm~\ref{alg:vertex-disjoint-cycles} is $O(m^2 \log m)$.

Finding a double-threading of $G$ from the collection of maximum-turn vertex disjoint simple cycles given by Algorithm~\ref{alg:vertex-disjoint-cycles} can be done in $O(m)$ time~\cite[Section 2.2]{demaine2024graph}. 
\end{proofsketch}

It remains an open question whether the general \threading problem can be solved efficiently in graphs with maximum degree $3$. Recall that \autoref{fig:lower-bound} depicts a graph whose minimum-turn threadings are far from perfect. 

\section{Other Special Cases}
\label{sec:special}
In this section, we give polynomial-time algorithms for several additional variants of \threading. 
We show that \exactlydouble is equivalent to the \tsp problem. This result paves the way for several insights: a fixed-parameter tractable algorithm contingent on the maximum degree of the graph, approximation algorithms, and a polynomial-time algorithm if the goal is simply to minimize the number of turns. Next, we give a polynomial-time algorithm to thread rectangular grids. Finally, we conclude our paper with a polynomial-time approximation for \threading~given bounded positive turn costs. 

\subsection{Exactly-Double Threading}
\label{sec:tsp}

Consider $n$ cities, with distance $d(i,j)$ between cities $i$ and $j$. The \tsp~problem (TSP) asks if there exists a tour of length at most some constant $C \geq 0$ that visits each city exactly once. 

\begin{theorem}
\label{thm:tsp}
    \exactlydouble~is equivalent to \tsp. 
\end{theorem}

\begin{proofsketch}
First, we sketch a reduction from TSP to \exactlydouble. Given an instance of TSP, we construct an instance $(G, \alpha)$ of \exactlydouble: Create $n$ distinct arms incident to a central vertex $c$ as in the proof of Lemma~\ref{lem:ham-cycle}. Set the turn costs between any two arms $a_i$ and $a_j$ about vertex $c$ to $d(i,j)$. Set all other turn costs to $1$. It follows that there exists a tour of length at most $C$ that visits each city exactly once if and only if there exists a threading of $G$ of cost at most $C+4n$. 

Let's now consider the reverse direction. We reduce \exactlydouble~to TSP: Leveraging the key observation that each junction graph must be precisely a cycle, we find the minimum-cost cycle at each junction $v$ by solving TSP on $v$ --- we assign a city to each edge incident to $v$, with distance $\alpha(uv,vw)$ between cities $uv$ and $vw$. Subsequently, we apply the polynomial-time procedure by Demaine, Kirkpatrick, and Lin~\cite{demaine2024graph} to construct a threading, using these minimum-cost cycles as junction graphs. 
\end{proofsketch}

A fixed-parameter tractable algorithm for \exactlydouble~follows directly from the proof of Theorem~\ref{thm:tsp}: 

\begin{corollary}
\exactlydouble~can be solved in polynomial time if the maximum vertex degree $\Delta$ of the graph is in $O(\log n)$. 
\end{corollary}

\begin{proofsketch}
Solving TSP on each junction can be achieved in $O(d(v)^2 + 2^{d(v)})$ time via the Bellman-Held algorithm~\cite{bellman1962dynamic,held1962dynamic}, so the runtime is polynomial as long as $d(v) \leq \Delta$ is logarithmic in $n$. 
\end{proofsketch}

We can further inherit results from TSP: If the turn costs at each junction form a metric space, that is, if they obey the triangle inequality, then we can approximate \exactlydouble~in polynomial time by solving metric TSP approximately at each junction. For instance, if all turn costs are either $1$ or $2$, then we can utilize the approximation algorithm for (1, 2)-TSP by Adamaszek, Mnich, and Paluch~\cite{adamaszek2018new} to obtain an $8/7$-approximation scheme that runs in $O(\sum_{v\in V(G)}d(v)^3 + m)$ time. More generally, we can apply the Christofides-Serdyukov algorithm~\cite{christofides2022worst} to obtain a $3/2$-approximation scheme that achieves the same time complexity. 

\begin{corollary}
     \textsc{Exactly-Double Threading} is in \P~if turn costs are in $\{0, 1\}$. 
\end{corollary}

\begin{proofsketch}
This result follows from the fact that we can solve TSP in a \emph{complete} graph $G$ with costs in $\{0,1\}$ in polynomial time. 
\end{proofsketch}

The hardness of \double~given turn costs in $\set{0,1}$ remains open.

\subsection{Rectangular Grid Graphs}
\label{sec:grid}

A \defn{rectangular grid graph} $G$ is a $w \times h$ lattice graph consisting of $w$ rows and $h$ columns of vertices, referred to as \defn{vertex rows} and \defn{vertex columns}, respectively. Each 90-degree turn at a vertex incurs a cost of $1$ while moving in a straight line incurs no cost. Below, we introduce a strategy for efficiently computing a threading $G$ with few turns. We first consider the following observations:

\begin{observation}
\label{obs:num_turns}
    Every vertex requires at least one turn to connect its junction graph.
\end{observation}

\begin{proof}
The junction graph of vertex $v$ is connected only if the threading traverses $v$ uniquely at least $d(v)-1$ times. Let's consider the implication of this requirement on the three types of vertices in $G$: (1) A degree-$2$ vertex, being a corner, necessitates a turn. (2) A degree-$3$ vertex permits one straight traversal. If this traversal is taken, then any different traversal must be a turn. (3) A degree-$4$ vertex permits two separate straight traversals, but an additional distinct traversal is required, thereby requiring a turn.
\end{proof}

\begin{observation}
\label{obs:even_turns}
Each vertex column contains at least $2\lceil h/2 \rceil$ turns, and each vertex row contains at least $2\lceil w/2 \rceil$ turns. 
\end{observation}

\begin{proof}
Our subsequent argument for vertex columns applies symmetrically to vertex rows. By Observation~\ref{obs:num_turns}, it suffices to argue that an even number of turns is required for each vertex column. When a thread enters a vertex column, it must exit the column either by returning the way it entered, inducing two turns, or leaving in the opposite direction, which results in either no turns if going straight or two turns if forming a staircase.
\end{proof}

It follows that: 

\begin{lemma}
\label{lemma:grid_lb}
A threading of $G$ with $w$ vertex columns and $h$ vertex rows has at least $2\lceil w/2 \rceil  \cdot 2\lceil h/2 \rceil$ turns if one of $w$ or $h$ is even, or both. 
\end{lemma}

\begin{proof}
If both $w$ and $h$ are even, then the threading has at least $w \cdot h = 2\lceil w/2 \rceil  \cdot 2\lceil h/2 \rceil$ turns by Observation~\ref{obs:num_turns}, as desired. If $w$ is odd and $h$ is even, then each vertex row necessitates at least $w+1$ turns by Observation~\ref{obs:num_turns} and \ref{obs:even_turns}, so the threading once again has a total of at least $(w+1)\cdot h = 2\lceil w/2 \rceil  \cdot 2\lceil h/2 \rceil$ turns. The same analysis applies to even $w$ and odd $h$. 
\end{proof}

\begin{theorem}
\label{theorem:grid_graph}
In $O(m)$ time, we can compute a threading of $G$ with $2\lceil w/2 \rceil \cdot 2\lceil h/2 \rceil$ if at least one of $w$ and $h$ are even, which is optimal by Lemma~\ref{lemma:grid_lb}, or $2\lceil w/2 \rceil \cdot 2\lceil h/2 \rceil-4$ turns if both $w$ and $h$ are odd. 
\end{theorem}

To prove Theorem~\ref{theorem:grid_graph}, we leverage the following lemma: 

\begin{lemma}
\label{lemma:cycle_to_threading}
Given an edge cycle cover $\mathcal{C}$ of $G$, where each cycle contains no U-turns and collectively traverses each $v \in V(G)$ uniquely at least $d(v)-1$ times, we can construct a threading of $G$ with cost $\sum_{c \in \mathcal{C}} \alpha(C)$. 
\end{lemma}

\begin{proofsketch}
The cycles in $\mathcal{C}$ traverse each edge of $G$ at least once, as required by the definition of an edge cycle cover, and since they uniquely traverse each junction at least $d(v)-1$ times without U-turns, they connect all junction graphs. Therefore, we can construct a threading of equal turn cost as $\mathcal{C}$ from the junction graphs induced by $\mathcal{C}$. 
\end{proofsketch}

We present Algorithm~\ref{alg:cycle_cover} for threading $G$ by finding such an edge cycle cover of $G$. 

\begin{algorithm}[H]
\caption{Threading a rectangular grid graph}
\label{alg:cycle_cover}
\begin{enumerate}
    \item Construct an edge cycle cover $\mathcal{C}$ of $G$ as follows:
    \begin{enumerate}[i.]
        \item 
        \label{step:case}
        If $w$ and $h$ are both even (Figure~\ref{fig:rectangle-graphs}a), then: 
        \begin{enumerate}[A.]
            \item 
            Add a cycle covering the boundary of $G$. 
            \item 
            Add a cycle covering each even row and each even column of $G$. 
            \item
            Add a cycle covering each square $(i, j)$ in $G$ for all pairs of even $i, j$. 
        \end{enumerate}
        \item If $w$ is even and $h$ is odd (Figure~\ref{fig:rectangle-graphs}b): 
        \begin{enumerate}[A.]
            \item Perform steps A and B from (i), excluding the cycle covering the last row of $G$. 
            \item Add a cycle covering the left, top, and right sides of $G$, with a jagged traversal of the bottom boundary that dips the even squares. 
        \end{enumerate}
        \item If both $w$ and $h$ are odd (Figure~\ref{fig:rectangle-graphs}c): 
        \begin{enumerate}[A.]
            \item Add a cycle covering each odd row and each odd column of $G$. 
            \item Add a cycle covering each square $(i, j)$ in $G$ for for all pairs of even $i, j$.
        \end{enumerate}
    \end{enumerate}
    \item Compute the junction graphs induced by $\mathcal{C}$.
    \item Compute a threading from the junction graphs (Section 2.2~\cite{demaine2024graph}). 
\end{enumerate}
\end{algorithm}

\begin{figure*}[h]
\centering
\includegraphics[width=\textwidth]{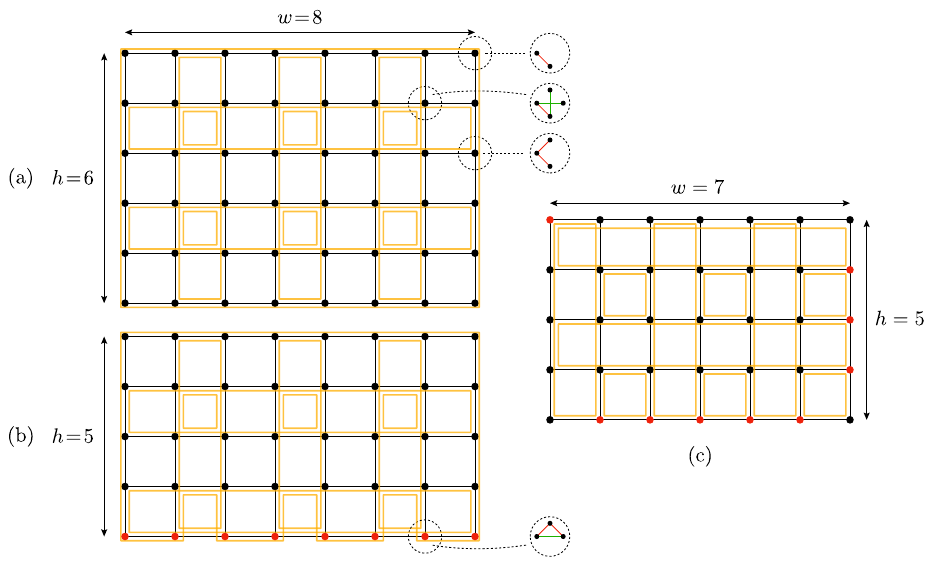}
\caption{\label{fig:rectangle-graphs} Illustration of the cycles created by Algorithm~\ref{alg:cycle_cover}. Vertices highlighted in red are traversed by two turns rather than one.}
\end{figure*} 

\subparagraph{Correctness.} It is straightforward to verify that the edge cycle cover $\mathcal{C}$ constructed by Algorithm~\ref{alg:cycle_cover} indeed covers each edge of $G$ and uniquely traverses each vertex $v$ at least $d(v)-1$ times. We demonstrate this verification for the first case (i) where both $w$ and $h$ are even. For edges: Every edge along the boundary of $G$ is covered by the cycle from Step A, and every interior vertical and horizontal edge is covered by cycles from steps B and C, respectively. For vertices: Every degree-2 junction is connected via a cycle from Step A (turn), every degree-3 vertex is connected via two cycles --- one from Step B (straight) and the other from Step C (turn), and every degree-4 vertex is connected via three cycles --- two from Step B (straight) and one from Step C (turn). By Lemma~\ref{lemma:cycle_to_threading}, we achieve a threading of $G$ with the same number of turns as $\mathcal{C}$, which in this case amounts to $2\lceil w/2 \rceil \cdot 2\lceil h/2 \rceil$. Each step of the algorithm operates in $O(m)$ time, thus completing our proof of Theorem~\ref{theorem:grid_graph}.

\subsection{Approximation for Bounded Turn Costs}
\label{sec:approximation}

\begin{theorem}
Given that all turn costs are positive, we can attain a $2r$-approximation for \threading, with $r$ representing the ratio between the maximum and minimum turn costs. This approximation can be computed in $O(m)$ time.
\end{theorem}

\begin{proof}
Apply the na\"ive double-threading algorithm described in Section~\ref{sec:preliminaries} to obtain a threading $T$. We bound $\alpha(T)$ as follows: Normalize all turn costs such that the smallest is $1$. Then the cost of $T$ at each junction $v$ must be at most $rd(v)$ (i.e., the maximum cost of a cycle at $v$), and the cost of an optimal threading \textsf{OPT} is at least $d(v)-1$ (i.e., the minimum cost of a tree at $v$). Hence, 
\[
\frac{\alpha(T)}{\alpha(\textsf{OPT})} \leq \frac{r \cdot \sum_{v \in V} d(v)}{\sum_{v\in V} d(v)-1} = \frac{2rm}{2m-n} \leq \frac{2rm}{m} = 2r
\]  
as desired.
\end{proof}

%% file: sections/5-conclusions.tex
\section{Future Work}

Interesting avenues for future work include further characterizing \threading in geometric graphs, such as by
\begin{itemize}
    \item Completing our analysis of rectangular grid graphs and expanding the study to general grid graphs with irregular boundaries and holes.
    \item Determining whether our results for \perfect extend to planar graphs, especially those whose turn costs respect the physical constraints inherent to geometries.\footnote{Planar \oneinthreesat is NP-hard~\cite{mulzer2008minimum}.}
\end{itemize}
Other open questions include the hardness of \threading in graphs of maximum degree 3 and the hardness of \double when turn costs are restricted to $0$ or $1$.

%% file: main.bbl
\begin{thebibliography}{10}

\bibitem{adamaszek2018new}
Anna Adamaszek, Matthias Mnich, and Katarzyna Paluch.
\newblock New approximation algorithms for $(1, 2)$-{TSP}.
\newblock In Ioannis Chatzigiannakis, Christos Kaklamanis, D\'{a}niel Marx, and Donald Sannella, editors, {\em 45th International Colloquium on Automata, Languages, and Programming (ICALP 2018)}, volume 107 of {\em Leibniz International Proceedings in Informatics (LIPIcs)}, pages 9:1--9:14, 2018.
\newblock \href {https://doi.org/10.4230/LIPIcs.ICALP.2018.9} {\path{doi:10.4230/LIPIcs.ICALP.2018.9}}.

\bibitem{aggarwal2000angular}
Alok Aggarwal, Don Coppersmith, Sanjeev Khanna, Rajeev Motwani, and Baruch Schieber.
\newblock The angular-metric traveling salesman problem.
\newblock {\em SIAM Journal on Computing}, 29(3):697--711, 2000.

\bibitem{arkin2005optimal}
Esther~M. Arkin, Michael~A. Bender, Erik~D. Demaine, S{\'a}ndor~P. Fekete, Joseph S.~B. Mitchell, and Saurabh Sethia.
\newblock Optimal covering tours with turn costs.
\newblock {\em SIAM Journal on Computing}, 35(3):531--566, 2005.

\bibitem{bellman1962dynamic}
Richard Bellman.
\newblock Dynamic programming treatment of the travelling salesman problem.
\newblock {\em Journal of the ACM}, 9(1):61--63, 1962.

\bibitem{bern2023fabrication}
James~M. Bern, Zach~J. Patterson, Leonardo~Zamora Ya{\~n}ez, Kristoff~K. Misquitta, and Daniela Rus.
\newblock A fabrication and simulation recipe for untethering soft-rigid robots with cable-driven stiffness modulation.
\newblock In {\em 2023 IEEE/RSJ International Conference on Intelligent Robots and Systems (IROS)}, pages 8337--8342. IEEE, 2023.

\bibitem{bern2022contact}
James~M. Bern, Leonardo~Zamora Ya{\~n}ez, Emily Sologuren, and Daniela Rus.
\newblock Contact-rich soft-rigid robots inspired by push puppets.
\newblock In {\em 2022 IEEE 5th International Conference on Soft Robotics (RoboSoft)}, pages 607--613. IEEE, 2022.

\bibitem{christofides2022worst}
Nicos Christofides.
\newblock Worst-case analysis of a new heuristic for the travelling salesman problem.
\newblock {\em Operations Research Forum}, 3(1):20, 2022.

\bibitem{de1998upper}
Dominique de~Caen.
\newblock An upper bound on the sum of squares of degrees in a graph.
\newblock {\em Discrete Mathematics}, 185(1-3):245--248, 1998.

\bibitem{demaine2024graph}
Erik~D. Demaine, Yael Kirkpatrick, and Rebecca Lin.
\newblock Graph threading.
\newblock In Venkatesan Guruswami, editor, {\em 15th Innovations in Theoretical Computer Science Conference (ITCS 2024)}, volume 287 of {\em Leibniz International Proceedings in Informatics (LIPIcs)}, pages 38:1--38:18, 2024.
\newblock \href {https://doi.org/10.4230/LIPIcs.ITCS.2024.38} {\path{doi:10.4230/LIPIcs.ITCS.2024.38}}.

\bibitem{ellis2015dna}
Joanna~A. Ellis-Monaghan, Andrew McDowell, Iain Moffatt, and Greta Pangborn.
\newblock {DNA} origami and the complexity of {Eulerian} circuits with turning costs.
\newblock {\em Natural Computing}, 14:491--503, 2015.

\bibitem{galil1986ev}
Zvi Galil, Silvio Micali, and Harold Gabow.
\newblock An {$O(EV \log V)$} algorithm for finding a maximal weighted matching in general graphs.
\newblock {\em SIAM Journal on Computing}, 15(1):120--130, 1986.

\bibitem{garey1979computers}
Michael~R. Garey and David~S. Johnson.
\newblock {\em Computers and Intractability: A Guide to the Theory of {NP}-Completeness}.
\newblock W. H. Freeman, 1979.

\bibitem{held1962dynamic}
Michael Held and Richard~M. Karp.
\newblock A dynamic programming approach to sequencing problems.
\newblock {\em Journal of the Society for Industrial and Applied mathematics}, 10(1):196--210, 1962.

\bibitem{Kyosev2015braiding}
Y.~Kyosev.
\newblock Yarn winding operations in braiding.
\newblock In {\em Braiding Technology for Textiles}, Woodhead Publishing Series in Textiles, chapter~10, pages 231--254. Woodhead Publishing, 2015.

\bibitem{Lin2024push}
Rebecca Lin and Tomohiro Tachi.
\newblock Push puppet-inspired deployable structure, Jan 2024.
\newblock URL: \url{https://twitter.com/rebeccayelin/status/1749193197031469102}.

\bibitem{mulzer2008minimum}
Wolfgang Mulzer and G{\"u}nter Rote.
\newblock Minimum-weight triangulation is np-hard.
\newblock {\em Journal of the ACM (JACM)}, 55(2):1--29, 2008.

\bibitem{patterson2023safe}
Zach~J. Patterson, Wei Xiao, Emily Sologuren, and Daniela Rus.
\newblock Safe control for soft-rigid robots with self-contact using control barrier functions.
\newblock arXiv:2311.03189, 2023.
\newblock URL: \url{https://arXiv.org/abs/2311.03189}.

\end{thebibliography}
